\documentclass[a4paper,10pt, twocolumn]{article}

\usepackage{algorithm}
\usepackage{algorithmic}

\usepackage{amssymb,amsmath,amsthm,url}

\topmargin by -\headsep

\DeclareMathOperator*{\E}{\mathbb{E}}

\DeclareMathOperator*{\Directprod}{\bigoplus}
\DeclareMathOperator*{\Union}{\bigcup}

\newcommand{\e}{\epsilon}

\newtheorem{theorem}{Theorem}[section]

\newtheorem{lemma}[theorem]{Lemma}

\newtheorem{claim}[theorem]{Claim}
\newtheorem{definition}[theorem]{Definition}

\newcommand{\polylog}{\mathrm{polylog}}

\newcommand{\set}[1]{\left\{#1\right\}}
\newcommand{\card}[1]{\left|#1\right|}

\title{On constant factor approximation for earth mover distance over doubling
metrics}
\author{
  Shi Li\\
  Dept. of Computer Science\\
  Princeton University\\
  35 Olden Street\\
  Princeton, NJ, 08540\\
  shili@cs.princeton.edu
}
\date{February 9, 2010}

\begin{document}

\maketitle

\begin{abstract}
Given a metric space $(X,d_X)$, the earth mover distance between two
distributions over $X$ is defined as the minimum cost of a bipartite matching
between the two distributions. The doubling dimension of a metric $(X, d_X)$ is
the smallest value $\alpha$ such that every ball in $X$ can be covered by
$2^\alpha$ ball of half the radius. A metric (or a sequence of metrics) is
called doubling precisely if its doubling dimension is bounded.  We study the
efficient algorithms for
approximating earth mover distance over doubling precisely metrics. 

Our first result is a near linear time (in the size of the $X$) algorithm for
estimating EMD over doubling metric $X$, with a $O(\alpha_X)$ approximation
ratio, where $\alpha_X$ is the doubling dimension of $X$. Given a metric $(X,
d_X)$, we can
use $\tilde O(n^2)$ preprocessing time to create a
data structure of size $\tilde O(n^{1 + \e})$, such that subsequent EMD queries
can be
answered in $\tilde O(n)$ time, with approximation ratio $O(\alpha_X/\e)$.

Our second result is an encoding scheme, which is a weaker form of sketching.
In an encoding scheme, distributions are encoded, such that the EMD
between two distributions can be estimated in \emph{sub linear} time, given the
encodings of the two distributions. In particular, given $(X, d_X)$, by using
$\tilde O(n^2)$ preprocessing time, every subsequent
distribution $\mu$ can be encoded into $F(\mu)$ in $\tilde O(n^{1 + \e})$ time.
The query for EMD between $\mu$ and $\nu$ can be answered in $\tilde O(n^\e)$
time, with approximation ratio $O(\alpha_X/\e)$, given the two encodings
$F(\mu)$ and $F(\nu)$.

The encoding scheme has immediate applications. In a 2-player game where 1
player knows $\mu$ and the other knows $\nu$, there is a communication protocol
with small communication complexity, through which the two players can
approximate the EMD between $\mu$ and $\nu$.   Another application is distance
oracle, where we are given a metric $(X, d_X)$ and $s$ distributions $\mu_1,
\mu_2, \cdots, \mu_s$ over $X$, we can use $\tilde O(n^2 + sn^{1+\e})$
preprocessing time, creating a data structure of size $\tilde O(n^{1 + \e} +
sn)$, such the query for EMD between $\mu_i$ and $\mu_j$ can be answered in
$\tilde O(n^{\e})$ time,  with approximation ratio $O(\alpha_X/\e)$.
\end{abstract}

\section{Introduction}

Given a finite metric $(X, d_X)$ and two multi-sets $A, B$ of points in $X$ with
$|A| = |B| = N$, the \emph{Earth-Mover distance}(or \emph{EMD}) between $A$ and
$B$ is defined as the minimum cost of a perfect matching, with respect to the
cost function $d_X$, i.e:
$$EMD_X(A, B) = \min_{\pi:A \to B}\sum_{p \in A}\|d_X(a, \pi(a)\|$$
where $\pi$ ranges over all one-to-one mappings.

The EMD metric is of significant importance in many applications. For example,
in computer vision, EMD is used as a measurement for dissimilarity between two
images. The idea is to represent an image as a distribution on features (such as
colors and color spectrum), with an underlying metric. Then, the EMD of two
distributions of features can be used for navigating image databases and image
retrieve(\cite{RGT97}, \cite{RTG98a}, \cite{RTG98b}, \cite{CG97}, \cite{CG99},
\cite{IT03}).

The exact distance can be computed in $O(N^3)$ time for general underlying
metrics, using Hungarian method(\cite{Law76}). For metric
supported by a sparse graph, there is a scaling algorithm due to\cite{GT89}
that runs in $O(\sqrt{|V|}|E|\log(|V|N))$ time. Better algorithms are known for
special metrics. EMD over a 2-dimensional plane can be computed in
$O(N^{5/2}\log^{O(1)}N)$ time, due to Vaidya\cite{Vai88}. This time is further
improved to $O(N^{2 + \delta})$, for any $\delta > 0$, by Agarwal etc.
\cite{AES95}.

Even for the special metric such as 2-dimensional plane, the exact EMD requires
super-quadratic time, too expensive in many applications.  This
computational bottleneck motivates faster approximation algorithms. For the
EMD over 2-dimensional plane, Agarwal and Varadarajan \cite{AV99} showed
$O(N^{3/2}/\e^2\log^5(N/\e))$ time $(1 + \e)$-approximation, then they gave an
improved algorithm with $N^{1 + \delta}\log^{O(1)}N$-time and
$\log(1/\delta)$-approximation(\cite{AV04}). Indyk \cite{Ind07} further
improved the running time to $O(N\log^{O(1)}N)$ for constant approximation.
Following \cite{Ind07}, Andoni etc. \cite{ABIW09} provided a sketching
algorithm with constant approximation.

There has also been special interest for embedding EMD into normed spaces.
Charikar \cite{Cha02} showed that EMD over $(X, d_X)$ can be embedded into
$l_1$ space with distortion $\alpha$, if $(X, d_X)$ can be embedded into
distribution of dominating trees with distortion $\alpha$. Then,
by \cite{FRT03}, which showed $\alpha$ is at most $O(\log n)$ for a $n$-point
metric, EMD over $(X, d_X)$ with can be embedded into $l_1$ with distortion
$O(\log n)$. Embedding into $l_1$ can give approximation algorithms, with the
ratio equal to the distortion of the embedding. However, the embedding of
EMD into $l_1$ has limitations : it has been showed in \cite{NS06} that
embedding of EMD over grid $[n]^2$ must incur a distortion of at least
$\Omega(\sqrt{\log n})$.

In this paper, we are interested in approximating EMD over more general metrics:
metrics with bounded doubling dimensions.  The doubling dimension of a metric
$(X, d_X)$ is the smallest $\alpha$ such that every ball in $X$ can be covered
by $2^{\alpha}$ balls of half the radius. This is a richer family of metrics
than the constant-dimension Euclidean space. It can be shown, for example, for
any fixed $p$, the doubling dimension of
$d$-dimensional $l_p$ is roughly $d$. On the other hand, \cite{GKL03} showed a
family of metrics $(G_k, d_k)$ with bounded doubling dimension, whose embedding
into $l_2$ must incur a distortion of $\Omega(\sqrt{\log|G_k|})$. For a finite
metric space $X$, $\alpha_X \leq \log |X|$. For many problems involving metrics,
better results are known if the metrics have bounded doubling dimension.
\cite{KLMN04} obtained an $O(\sqrt{\alpha_X\log n})$ distortion embedding of
$X$ into $l_2$, an improvement over Bourgain's theorem(\cite{Bou85}) if the
metric has low doubling dimension. \cite{AFHK04} give better approximation
algorithms for metric labeling and 0-extension for $\alpha$-decomposable
metrics(if a metric has doubling dimension $\alpha$, it is
$O(\alpha)$-decomposable).

\subsection{Our contribution}

We study approximation algorithm for EMD over doubling metrics, a generalization
of low dimensional Euclidean spaces. As far as we know, this is the first
work to consider the EMD over this family of metrics. Our algorithm is a
generalization of the algorithm in \cite{Ind07} for approximating EMD over
planar grids. We defined a metric, called ``sibling-linked hierarchical
well-separated tree'' metric, that the doubling metric can be embedded into with
constant distortion. This metric has a nice property that the EMD over it is the
sum of EMD over smaller metrics, which allows us to do ``importance sampling''.
In this paper, we also promoted a scheme called ``encoding scheme'', a weaker
form of the sketching scheme. In the encoding scheme, distributions are encoded
in some form, and the EMD of two distributions can be approximated in
\emph{sub-linear} time, if the encodings of the two distributions are given.
The sub-linear time estimation algorithm does the importance sampling in a
binary-search way. using only logarithmic time.

Compared to the Euclidean spaces, there
is an issue for doubling metrics on how the underlying metric is given. Reading
the whole metric requires time quadratic in $n$, the size of the metric, while
we're aiming for algorithms with time near linear in $n$. We avoid this
bottleneck by preprocessing, in which our algorithms read the metric $(X, d_X)$
and create a data structure of small size. By doing this, subsequent queries
for EMD between $\mu$ and $\nu$ can be answered in time near linear in $n$.

Throughout, we use $\alpha_X$ to denote the doubling dimension of $X$. We assume
each coordinate in a distribution can be represented in $\polylog(n)$ bits. In
particular, for a fixed constant number $a$, we define 
$$\mathcal{P}_X = \set{\mu:X \to \set{0, \frac{1}{n^a}, \frac{2}{n^a}
,\cdots, 1} \Big|\sum_{p\in X}\mu_p = 1}$$
This restriction is only for simplicity of demonstration and doesn't affect our
algorithms.

Our first result is an almost linear time, constant approximation algorithm for
EMD over doubling metrics.

\begin{theorem}[Approximation algorithm]
 \label{Theorem:approximation-algorithm-1}
Let $(X, d_X)$ be a metric space with $|X| = n$. 
Let $0 < \e < 1$ be fixed. Given $(X,d_X)$ and $\alpha_X$, there
is an algorithm which, by using $\tilde O(n^2)$ preprocessing
time to create a data structure of size $\tilde O(n^{1+\e})$, for any
subsequent query for EMD between two distributions $\mu, \nu \in \mathcal{P}_X$,
outputs in $\tilde O(n)$ time a random estimation $D$ satisfying
\begin{equation}
\label{Equation:D-approximates-EMD}
 EMD_X(\mu, \nu)\leq D \leq O\left(\frac{\alpha_X}{\e}\right) EMD_X(\mu, \nu)
\end{equation}
with probability at least 2/3. The
probability is over the randomness for the preprocessing as well as the
estimation.
\end{theorem}

The preprocessing time is quadratic in the size of the metric, which is
unavoidable since we have to read the whole metric. However, the
algorithm only reads the metric once, after which
subsequent EMD queries can be answered in almost linear time in $n$.  The
probability $2/3$ can be amplified to $1 -
\varepsilon$, by repeating the algorithm $O(\log (1/\varepsilon))$ times.

Unlike \cite{ABIW09}, our algorithm for theorem
\ref{Theorem:approximation-algorithm-1} is not sketching-based. Instead, we
provide a weaker scheme, which we call ``encoding scheme". In the
encoding scheme,  distributions are encoded to $l_1$ vectors by an
encoding function $F$. Given two encodings $F(\mu)$ and $F(\nu)$, there is a
\emph{sub
linear} time algorithm approximating $EMD_X(\mu, \nu)$. The difference
between an encoding and a sketch is, an encoding is not necessarily shorter than
its pre-image.

\begin{theorem}[Encoding scheme]
\label{Theorem:encoding-scheme}
Let $X, d_X, n, \e$ be as in theorem \ref{Theorem:approximation-algorithm-1}.
Given $(X, d_X)$ and $\alpha_X$, there is an algorithm, which uses $\tilde
O(n^2)$
preprocessing time to create a data structure of size $\tilde O(n^{1 + \e})$,
such that for any subsequent query for $EMD_X(\mu, \nu)$, it can perform the
following 3 steps:
\begin{enumerate}
\item computes an encoding $F(\mu)$ of size $\tilde O(n)$ for $\mu$ in
$\tilde O(n^{1 + \e})$ time, only reading $\mu$ and the data structure;
\item computes $F(\nu)$ similarly;
\item outputs a random number $D$ satisfying (\ref{Equation:D-approximates-EMD})
with
probability 2/3, in $\tilde O(n^\e)$ time, only reading $F(\mu), F(\nu)$ and the
data structure.
\end{enumerate}
\end{theorem}

Notice that compared to the algorithm in theorem
\ref{Theorem:approximation-algorithm-1}, the algorithm in theorem
\ref{Theorem:encoding-scheme} requires $O(n^{1 + \e})$ time to approximate the
EMD. The encoding scheme implies the following two theorems.

\begin{theorem}[Communication protocol]
\label{Theorem:communication-protocol}
Let $X, d_X, n, \e$ be as in theorem
\ref{Theorem:approximation-algorithm-1}. Consider a game with two players
$Alice$ and $Bob$, in which $Alice$ knows $\mu \in \mathcal{P}_X$
and $Bob$ knows $\nu \in \mathcal{P}_X$. There
exists a communication protocol with complexity $\tilde
O(n^\e)$, through which $Alice$ and $Bob$ can output a number $D$ satisfying
(\ref{Equation:D-approximates-EMD}) with probability at least 2/3.
\end{theorem}

Theorem \ref{Theorem:communication-protocol} suggests that proving a good
lower bound for sketching using general communication complexity is impossible.

\begin{theorem}[Distance oracle]
\label{Theorem:distance-oracle}
Let $X, d_X, n, \e$ be stated as in theorem
\ref{Theorem:approximation-algorithm-1}. Given $(X, d_X)$, $\alpha_X$ and $s$
distributions $\mu_1, \mu_2, \cdots, \mu_s \in \mathcal{P}_X$, there is an
algorithm which
uses preprocessing time $\tilde O(n^2 + sn^{1+\e})$ to construct a data
structure of size $\tilde O(n^{1 + \e} + sn)$, and for any subsequent
query for EMD between
$\mu = \mu_i$ and $\nu = \mu_j$, outputs a number $D$ in $\tilde O(n^\e)$ time, 
satisfying (\ref{Equation:D-approximates-EMD}) with probability at least 2/3.
\end{theorem}

\subsection{Preliminaries}
 The doubling dimension $\alpha_X$ of a metric $(X, d_X)$ is defined as 
the minimum $t$ such that every ball in $X$ can be covered 
using $2^t$ balls of half the radius.  Define $R_X = \min{p\ in X} \max{q
\in X} d_X(p, q)$ to be the radius of $X$. When $X$ is clear from the context,
we may omit the subscript, using $\alpha, R$ instead. For a point $p \in X$ and
a real number $r$, we use $Ball(p, r) = \set{q : d_X(p ,q)\leq r}$ to denote the
ball of radius $r$ centered at $p$. 

For two distributions $\mu, \nu \in \mathcal{P}_X$, the earth mover
distance (EMD) between $\mu$ and $\nu$ is defined as 
$$EMD_X(\mu, \nu) = \min_{\pi : X \times X \to \mathbb{R}}\sum_{p, q \in
X}d_X(p, q)\pi(p, q)$$
where $\pi$ ranges over all functions $\pi$ satisfying 
\begin{eqnarray*}
 & & \forall p, q\in X, \pi(p, q) \geq 0;\\
 & & \forall p \in X, \sum_{q \in X}\pi(p, q) = \mu_p;\\
 & & \forall q \in X, \sum_{p \in X}\pi(p, q) = \nu_q.
\end{eqnarray*}

We'll use $l_1$ to denote the set of $L_1$ vectors of finite dimension, i.e
$l_1 = \bigcup_{i=1}^{\infty}\mathbb{R}^i$, equipped with $L_1$ norm. We'll use
$\oplus$ to denote the direct product operation. So, $x_1 \oplus x_2$ is the
direct product of $x_1$ and $x_2$,  and $\Directprod_{i=1}^s x_s$ is
the direct product $x_1, x_2, \cdots x_s$. 

The Cauchy distribution $\mathbb{C}(x_0, \gamma)$ is the distribution with the
following probability density function :
$f(x) = \frac{1}{\pi}\left(\frac{\gamma}{(x - x_0)^2 + \gamma^2}\right)$, where
$\gamma$ is called the scale parameter of the distribution. The cumulative
distribution function of $\mathbb{C}(x_0, \gamma)$ is $F(x) =
\frac{1}{\pi}\arctan\left(\frac{x - x_0}{\gamma}\right) + \frac12$. Cauchy
distribution is 1-stable distribution, meaning that the sum of $n$ independent
variables from $\mathbb{C}(0, 1)$ is $\mathbb{C}(0, n)$.

Throughout, we'll use $\tilde O()$ notion to hide a factor of $\polylog(n)$.

\section{Overview of the algorithms}
\label{Section:overview}

Our algorithms for theorem \ref{Theorem:approximation-algorithm-1} and
theorem \ref{Theorem:encoding-scheme} are only slightly different.  We'll
describe the major components of the algorithms here, and give the details in
following sections.

In the preprocessing, we read $(X, d_X)$ and decompose
it into a distribution of so called \emph{sibling-linked
hierarchical well separated trees}, or SLHSTs, with distortion
$O(\alpha/\e)$, where $\e$ is some parameter. The name is a little confusing
since SLHST is actually not a tree (recall that even grids can not be embedded
into distribution of trees with
distortion $o(\log n)$).  A SLHST is constructed by adding links
connecting children of the same vertex (or siblings) to a base HST.  In a
HST,
there is a unique path connecting two leaves. While in SLHST, we take a shortcut
for the this path: the two children of the last common
ancestor of the two leaves are directly connected. A perfect analogy for this
metric is the postal service
system. In the postal service system with a hierarchy of post offices, a package
is sent from a terminal to a local post office, and then to one with higher
level in the hierarchy, and then to a even higher level, until it can be sent
to an office of the same level. Then the package is sent to lower
levels along the hierarchy system, until it reaches the destination.

Our embedding into distribution of SLHSTs is dominating, and has small
distortion. Then, by a similar argument as the one in \cite{Cha02},
the EMD over $X$ can be approximated by the EMD over a
random SLHST from the distribution. A nice property
that a SLHST has is, the EMD over it can be
represented as a sum of $O(n)$ EMDs over smaller sub metrics. Recall a special
case in
\cite{Ind07}, the EMD over a grid of size $n$ is the sum of EMDs over grids
of smaller size. 

The decomposition into smaller EMDs allows us to use the ``importance sampling"
technique, which is also used in \cite{Ind07}. Suppose we want to compute a sum
$Z = \sum_{i=1}^N Z_i$, but it's too expensive to compute all the $Z_i$s. We can
choose a term $Z_i$ with probability roughly proportional to how much it
contributes to the sum. Then $Z_i$ divided by this probability is used as the
estimation for $Z$. The expectation of the value is $Z$, while the deviation
depends on how well the sampling distribution approximates the weight
distribution. 
Applying importance sampling to our case, we need a technique to approximate
each EMD over a small metric up to a reasonable factor. This can be
done using the
embedding of metrics to distribution of dominating trees and then computing the
EMD over the trees(\cite{FRT03}, \cite{Cha02}).

The above techniques provided us the main elements needed to prove theorem
\ref{Theorem:approximation-algorithm-1}, while not enough to prove theorem
\ref{Theorem:communication-protocol}. To get a sub-linear time estimation
algorithm, we need to do the binary sampling more carefully. We invented
a ``binary importance sampling" process, which allows us to do the
importance sampling in logarithmic time.  The high level idea is to sample in a
binary-search way. Recall in an importance sampling, we want to compute $Z =
\sum_{i = 1}^N Z_i$. We maintain a possible set $S$ of terms, initially $[N]$.
Each time the set $S$ is divided into two equal subsets $S_1$ and $S_2$, and
then replaced with one of two subsets, with probability proportional
to the weight of the subset. In $\log N$ steps, we end up with a single term in
$S$, which is the term selected by the algorithm.  This method requires us to
estimate the sum of a subset of terms. In our case, each term can be
approximated by the $L_1$ norm of a vector, so the total weight of a subset
is the sum of $L_1$ norms, which is equivalent to one $L_1$ norm.
We can use the sketch scheme in \cite{KOR98} to approximate the $L_1$
norm. For every potentially
possible set $S$, we have a sketch for the concatenation of vectors in
$S$. The sketch is linear, a crucial property allowing us to encode two
distributions
separately. We'll describe the encoding function and the binary sampling
algorithm in
section \ref{Section:binary-importance-sampling}.

\section{Sibling-linked hierarchical well-separated trees}

We introduce the definition of $k$-HST from \cite{Bar96} and then based on this
definition, we define a $k$-SLHST.
\begin{definition}[$k$-Hierarchical well-separated tree(\cite{Bar96})]
A $k$-hierarchical well-separated tree($k$-HST) is defined as a rooted
weighted
tree with the following properties:
\begin{enumerate}
\item The edge weight from any node to each of its children is the same;
\item The edge weights along any path from the root to a leaf are decreasing by
a factor of at least $k$.
\end{enumerate}
\end{definition}

We call $k$ the scale-decreasing factor of the HST.
For a tree
$k$-HST $\tau$, we use $V_\tau, U_\tau, deg(\tau)$ and $dep(\tau)$ to denote
$\tau$'s vertices, non-leaf vertices, degree and depth, respectively. 
For some $v \in U_\tau$, let $\Lambda(v)$ be the set of $v$'s children,
$\Gamma(v)$ be the set of $v$'s offspring leaves
and $\Delta_v$ be the length of edges from $v$ to its children.
 
\begin{definition}[$k$-sibling-linked hierarchical well-separated tree]
Let $\tau$ be a $k$-HST. We associate each $v \in U_\tau$ a 
metric $d_v : \Lambda_v \times \Lambda_v \to \mathbb{R}$ satisfying 
\begin{eqnarray}
\forall u,w\in \Lambda(v), d_v(u, w) \leq 2\Delta_v
\label{Equation:distance-requirement}
\end{eqnarray}
For some $v\in U_\tau$ and $u, w \in \Lambda(v)$, a \emph{sibling
link} $(u, w)$
is an edge of length $d_v(u,w)$ connecting
$u$ and $w$. We call the graph obtained by combining $\tau$ and all the sibling
links a \emph{$k$-sibling linked
hierarchical
well-separated tree}($k$-SLHST). 
\end{definition}

For a $k$-SLHST $T$, we use $d_T$ to denote its shortest path metric. 
 The notions such as $dep(\tau)$, $deg(\tau)$, $\Delta_v$, $\Lambda_v$ and
$\Gamma_v$ are naturally extended to $k$-SLHSTs. We say a SLHST $T$ supports
$X$ if $X$ is the leaves of $T$. 

To avoid confusion, for a SLHST $T$ that supports $X$, we always use $p, q$ to
denote points in $X$, as well as the leaves of $T$, and $u, v$ to denote inner
vertices in a SLHST, if possible. 

\subsection{Embedding $X$ into a distribution of SLHSTs}

Now, we are going to describe our algorithm for embedding metric $(X, d_X)$ to
a distribution of dominating SLHSTs. The algorithm is almost the same as the
HST embedding algorithm in \cite{FRT03}, except that we use $n^{1/\e}$ as the
scale-decreasing factor, instead of 2. The embedding
is pretty bad(it will incur a distortion of $O(n^\e)$). As we'll show, after we
insert the sibling-links to the HST, the distortion becomes very small. We
partition a set of radius $R$ to clusters of radius $r$ in $\log (R/r)$
steps and in each step we reduce the radius by a factor of 2. 

Despite of its similarity with the tree embedding algorithm of \cite{FRT03}, we
give our algorithm here, for the integrality of the paper. Algorithm
\ref{Algorithm:partition} partitions a metric $Y$ in to sets of smaller radius,
and it will be used in by algorithm \ref{Algorithm:embedding-into-slhst}.

\begin{algorithm}
 \caption{$partition(Y, d_Y, r)$}
\textbf{Input:} A metric $(Y, d_Y)$ and a scale $r$, $r$ is guaranteed to be
at least $R_Y/2$;\\
\textbf{Output:} A partition of $Y$ into clusters of radius smaller than $r$ :
$\set{Y_i : i \in [s]}$
\label{Algorithm:partition}
\begin{algorithmic}[1]
 \STATE Select a $r/2$-net $S \subset Y$;
\STATE Randomly choose a permutation $\pi$ for $S$ and a number $\beta \in
[1,2)$
\FOR {$i = 1, 2, \cdots,|S|$}
\STATE $Y_i \leftarrow Ball_{Y}\left(\pi(i), \beta r/2\right)\backslash
\Union_{j = 1}^{i-1}Y_j$; 
\ENDFOR
\RETURN $\set{Y_i: i \in [|S|], Y_i \neq \emptyset}$
\end{algorithmic}
\end{algorithm}

\begin{claim}
\label{Claim:properties-for-partion}
 The returned $\set{Y_i : i \in s}$ from algorithm
\ref{Algorithm:partition} is indeed a partition, i.e, 
$\forall i\neq j, Y_i \cap Y_j = \emptyset$ and $\bigcup_{i\in [s]}Y_i = Y$.
\end{claim}

\begin{claim}
\label{Claim:s-is-small} 
The number of clusters that algorithm \ref{Algorithm:partition} returns is
at most $(4R_Y/r)^{\alpha_Y} \leq 8^{\alpha_Y}$, for $r\geq R_Y/2$.
\end{claim}

Algorithm \ref{Algorithm:embedding-into-slhst} describes how a SLHST is
constructed. The tree is essentially a laminar tree of the $X$, where each
vertex in the tree is correspondent to a subset of $X$. A vertex $v_Y$ at
level $i$ (the root has level 0) has $\Delta_{v_Y} = r_i$, and is
associated with a metric $d_{v_Y}(v_{Y'}, v_{Y''}) = d_X(c_{Y'}, c_{Y''}) +
4r_{i+1}$, where $c_{Y'}$ and $c_{Y''}$ are the centers chosen for
$Y'$ and $Y''$ respectively. The additive term $4r_{i+1}$ is to
make sure that the SLHST metric is dominating.

\begin{algorithm}[ht]
 \caption{$embedding\_into\_SLHST(X, d_X, \e)$}
 \label{Algorithm:embedding-into-slhst}
 \textbf{Input:} A metric $(X, d_X)$ and $0 < \e \leq 1/3$ such that $1/\e$ is
an integer;\\
 \textbf{Output:} A SLHST $T$;
 \begin{algorithmic}[1]
  \STATE $h \leftarrow \log R_X$;
  \STATE $\mathcal{C}\leftarrow \set{X}, \mathcal{C}' \leftarrow
\emptyset, V_{\tau} = \set{v_X}, E_{\tau} \leftarrow \emptyset$;
 \FOR{$i \leftarrow 1, \cdots, h$}
 \FOR{Every $Y \in \mathcal{C}$}
  \STATE $\mathcal{C}'' \leftarrow partition(Y, d_{X|Y},
2^{h-i})$;
  \FOR{Every $Y' \in \mathcal{C}''$}
     \STATE $V_{\tau}\leftarrow V_{\tau} + \set{Y'},
E_{\tau}\leftarrow E_{\tau} + \set{(v_Y, v_{Y'})}$
  \ENDFOR
  \STATE $\mathcal{C}' \leftarrow \mathcal{C}' + \mathcal{C}''$;
   \ENDFOR
   \STATE $\mathcal{C} \leftarrow \mathcal{C}', \mathcal{C}'\leftarrow
\emptyset$;
   \ENDFOR
  \STATE $a \leftarrow \lceil\e \log n/\alpha_X\rceil$, $b \leftarrow
$ random
integer from $[a]$;
  \STATE $L = \set{i: 0\leq i \leq h, i \equiv b(\mbox{mod } a)} + \set{0, h}$;
  \STATE Shrink $\tau$ at level set $L$: we remove all level-$i$ vertices
for all $i\notin L$, and directly connect level-$j$ vertices to their level-$i$
ancestors using edges of length $2^{h - i}$, for $i < j$ and $i,j$ adjacent
numbers in $L$; let $\tau'$ be the new tree.
  \STATE Let $i<j$ be two adjacent numbers in $L$, for a level-$i$ vertex $v_Y$
and two level-$j$ vertices $v_{Y'}, v_{Y''}$ whose parent in $\tau'$ is
$v_Y$(so, $Y', Y'' \subset Y$), $d_{v_Y} = d_X(c_{Y'}, c_{Y''})$,
where $c_{Y'}$ is the center of $Y'$(i.e, the radius of $Y'$ is the
maximum distance from $c_{Y'}$ to some other point in $Y'$).
   \RETURN the SLHST $T$ defined by $\tau'$ and $d_{v_Y}$s.
 \end{algorithmic}
\end{algorithm}

To avoid confusion, we use ``rank'' instead of ``level'' to denote the positions
of vertices in $\tau, \tau'$ and $T$. We say a vertex $v \in
v_\tau$ has ``rank'' $i$, if the path from $v$ to the root in $\tau$ contains
$i$ edges. The rank of a vertex in $T$(or $\tau'$) is just the rank of its
correspondent vertex in $\tau$.
Notice that if some vertex $v\in U_T=U_{\tau'}$ has rank $i$, then $i \in L$.

\begin{lemma}
The algorithm \ref{Algorithm:embedding-into-slhst} actually returns a SLHST,
i.e, the associated metrics satisfy equation
(\ref{Equation:distance-requirement}).
Furthermore, $deg(T)\leq n^{O(\e)}$.
\end{lemma}
\begin{proof}
Let $v_Y \in U_T$ be a rank $i$ vertex and $v_{Y'}, v_{Y''}$ be two of its
chilren in $T$. Notice that $v_{Y'}, v_{Y''}$ has the same rank $j$, where $i$
and $j$ are adjacent in $L$.  Then $d_{v_Y}(v_{Y'}, v_{Y''}) = d_X(c_{Y'},
c_{Y''})\leq 2R_Y \leq 2^{h-i+1} = 2\Delta_{v_Y}$.

The degree of $\tau$ is $2^{O(\alpha)}$, by claim \ref{Claim:s-is-small}. The
shrinking operation collapses $a$ levels into 1 level, and thus, the degree
becomes at most $2^{O(\alpha)a} = 2^{O(\alpha)\lceil\e\log n/\alpha\rceil} =
n^{O(\e)}$.
\end{proof}

\begin{lemma} \label{Lemma:embedding-into-SLHSTs}
 The random SLHST $T$ that algorithm \ref{Algorithm:embedding-into-slhst}
returns supports $X$ and satisfies: 
\begin{enumerate}
\item $\forall p,q \in X, d_T(p, q) \geq d_X(p,q)$;
\item $\forall p, q\in X, \E_{T}\left[d_T(p,
q)\right] \leq O(\alpha/\e) d_X(p, q)$.
\end{enumerate}
\end{lemma}

\begin{proof}
Fix two points $p$ and $q$. Let $v_Y, v_{Y'}$ be the two highest-rank vertices
in the shortest path between $p$ and $q$ in $T$, and $v$ be their parent. Let
the rank of $v_Y$ and $v_{Y'}$ be $i$. If $i = h$, clearly $d_T(p,q) =
d_X(p,q)$; if $i < h$, 
\begin{eqnarray*}
& & d_T(p, q)\geq d_v(v_Y, v_{Y'})\\
 &=& d_X(c_Y, c_{Y'}) + 2^{h-i+1} \\
 &\geq& d_X(p, q) - d_X(p, c_Y) - d_X(q, c_{Y'}) + 2^{h-i+1}\\
 &\geq& d_X(p, q) - 2^{h - i} - 2^{h - i} + 2^ {h - i + 1}\\
 &=& d_X(p,q);
\end{eqnarray*}
So, we've proved the first property.

Let $P_i$ be the probability that $p$ and $q$ are separated at rank $i$
in $\tau$. Obviously, $\sum_{i=1}^{h} P_i = 1$, and from \cite{FRT03}, we have
\begin{eqnarray*}
\sum_{i=1}^{h}P_i2^{h-i} &\leq& O(\log n)d_X(p, q)
\end{eqnarray*}

Now, we fix a tree $\tau$. Suppose $p$ and $q$ are separated at
rank $i$ in $\tau$. Let $v_{Y_j}$ and $v_{Y'_j}$ be the rank $j$ ancestors of
$p$ and $q$, respectively. If $i \leq h - a + 1$, $p$ and $q$ may be separated
at rank $i, i+1, \cdots, i+a-1$ in $\tau'$, each with probability $1/a$. The
expected $d_T(p, q)$, over all possible $b$s, denoted by $D_i$, is at most
\begin{eqnarray*}
& &\frac1a\sum_{j = i}^{i+a-1}\left(d_X(c_{Y_j}, c_{Y'_j})+
2\sum_{j' = j}^{h-1}2^{h-j'}\right)\\
&\leq& \frac1a\sum_{j =
i}^{i+a-1}\left(d_X(p,q)+2^{h-j}+2^{h-j}+4\times2^{h-j}\right)\\
&\leq& d_X(p, q) + \frac{12}{a}2^{h-i}
\end{eqnarray*}

If $p$ and $q$ are separated at rank $i$ in $\tau$, for some $i > h - a + 1$,
$D_i$ is at most
\begin{eqnarray*}
 & &\frac1a\left(\sum_{j=i}^{h}d_X(c_{Y_j},c_{Y'_j}) +
(i+a-1-h)d_X(p, q)\right)\\
&\leq&d_X(p, q) +
\frac1a\sum_{j=i}^{h}\left(2^{h-j}+2^{h-j}+2\sum_{j'=j}^{h-1}2^{h-j'}\right)\\
&\leq& d_X(p, q) + \frac{12}{a}2^{h-i}
\end{eqnarray*}

$\E[d_T(p, q)]$ Taking the expectation over $\tau$ and $b$, we get
\begin{eqnarray*}
 \E_{T}\left[d_T(p,q)\right]&=&\sum_{i=1}^{h}P_iD_i \\
 &\leq& \sum_{i=1}^{h}P_i\left(d_X(p, q) + \frac{12
}{a}2^{h-i}\right)\\
&\leq&\sum_{i=1}^{h}P_id_X(p,q)+\frac{12}{a}\sum_{i=1}^{h}P_i2^{h-i}\\
&\leq& d_X(p, q) + \frac{12}{a}O(\log n)d_X(p, q)\\
&\leq& O(\alpha/\e)d_X(p, q)
\end{eqnarray*}

The last inequality comes from the fact that $a = \lceil\e\log n/\alpha\rceil$.
\end{proof}

\subsection{Decomposition of EMD over a SLHST}
Before showing the main lemma of this subsection, we introduce some notions.

Let $T$ be a SLHST that supports $X$. For a distribution $\mu \in \mathcal{P}_X$
and a vertex $v \in U_T$, define $\mu_v = \sum_{p \in
\Gamma(v)}\mu_p$ and $\hat \mu_v = \Directprod_{u \in
\Lambda(v)}\mu_u$. Define extended EMD (or EEMD) between $\hat \mu_v$ and $\hat
\nu_v$ to be
\begin{eqnarray*}
EEMD_v(\hat \mu_v, \hat \nu_v) &=&
\min_{\pi : \Lambda_v \times \Lambda_v \to
\mathbb R} EEMD_v^\pi(\hat\mu_v, \hat\nu_v)\\
EEMD_v^\pi(\hat\mu_v, \hat\nu_v) &=&
\sum_{u, w \in \Lambda_v} \pi(u, w)d_v(u, w)\\ 
&+& \Delta_v\left(\sum_{u \in
\Lambda_v}\mu_u - \sum_{u, w\in \Lambda_v}\pi_{u,w}\right) \\
&+&\Delta_v\left(\sum_{w \in \Lambda_v}\nu_w  - \sum_{u, w\in
\Lambda_v}\pi_{u,w}\right)
\end{eqnarray*}
where $\pi$ ranges over all transportation functions satisfying
\begin{eqnarray*}
& &\forall u, w \in \Lambda_v, \pi(u,w)\geq 0;\\
& &\forall u \in \Lambda_v, \sum_{w \in \Lambda_v}\pi(u, w) \leq \mu_u;\\
& &\forall w \in \Lambda_v, \sum_{u \in \Lambda_v}\pi(u, w) \leq \nu_w.
\end{eqnarray*}

The definition says that all the unmatched units must be sent to $v$.
\begin{lemma}
\label{Lemma:emd-to-sum-emds}
$$EMD_T(\mu, \nu) = \sum_{v\in U_T}EEMD_v(\hat \mu_v, \hat \nu_v)$$
\end{lemma}

\begin{proof}[Proof sketch]
We can view distributions $\mu$ and $\nu$ as supplies and demands on the leaves
of the tree. The allowed operation is moving $\e$ amount of supplies (or
demands) along some edge $(u, v)$, which costs $\e l(u,v)$. We can cancel out
the same amount of supplies and demands on the same vertex for free. We can
further restrict the moving direction, so that the moves can not go
downwards.

We show that the best strategy is to match supplies and demands locally. i.e,
it can not move supplies along $(u,v)$ while moving demands along
$(w,v)$, for some $v\in U_T$ and $u, w\in \Lambda_v$. Otherwise we would moving
some amount of supplies directly from
$u$ to $w$ and then match the same amount of demands at $w$. This will 
make the cost smaller, since $d_v(u,w) < \Delta_v + \Delta_v$. Thus, in the best
matching, each vertex $v$ will receive $\max(\mu_v - \nu_v, 0)$ amount of
supplies and $\max(\nu_v - \mu_v, 0)$ amount of demands and cancels out all but
$\card{\mu_v -\nu_v}$ amount of supplies or demands, which will be sent to its
parent. This is exactly the sum of $EEMD_v$s over all $v \in U_T$. 
\end{proof}

\begin{lemma} \label{Lemma:emd-of-dominating-metrics}
If $X$ can be embedded into a set $\mathcal{Y}$ of dominating metrics, with
distortion $\beta$, then, $EMD_X$ can be embedded into $\set{EMD_Y : Y\in
\mathcal{Y}}$, with distortion $\beta$.
\end{lemma}

\begin{proof}
It's easy to see that $EMD_Y$ dominates $EMD_X$. 
Let $\pi : X\times X \to \mathbb{R}$ be the best transportation function for
$EMD_X(\mu,\nu)$, and $\psi : Y \times Y \to $ be the best
transportation function for $EMD_Y(\mu, \nu)$. By the definition of EMD, we
have 
\begin{eqnarray*}
EMD_X(\mu, \nu) &=& \sum_{p, q \in X}\pi(p, q)d_X(p, q)\\
EMD_Y(\mu, \nu) &=& \sum_{p, q \in X}\psi_Y(p, q)d_Y(p, q)
\end{eqnarray*}

Thus,
\begin{eqnarray*}
& &\E_{Y \in \mathcal{Y}}EMD_Y(\mu, \nu)\\
&=&\E_{Y \in \mathcal{Y}}\sum_{p, q
\in X}\psi_Y(p, q)d_Y(p, q)\\
&\leq&\E_{Y \in \mathcal{Y}}\sum_{p, q
\in X}\pi(p, q)d_Y(p, q)\\
&=&\sum_{p, q\in X}\pi(p,q)\E_{Y \in \mathcal{Y}}d_Y(p, q)\\
&\leq&\sum_{p, q\in X}\pi(p, q)\beta d_X(p,q)\\
&\leq& \beta EMD_X(\mu,\nu)
\end{eqnarray*}
\end{proof}

\begin{lemma} \label{Lemma:emd-of-dominating-SLHSTs}
If $\mathcal{T}$ be the distribution of
dominating SLHSTs that $O(\alpha_X/\e)$-approximate $X$, then
\begin{eqnarray*}
EMD_X(\mu, \nu) &\leq& \E_{T \sim
\mathcal{T}}\left(\sum_{v\in
U_T}EEMD_{\Lambda_v}(\mu_v, \nu_v)\right)\\ &\leq& O(\alpha_X/\e)
EMD_X(\mu,\nu)
\end{eqnarray*}
\end{lemma}
\begin{proof} 
This lemma immediately follows lemma \ref{Lemma:emd-to-sum-emds} and lemma
\ref{Lemma:emd-of-dominating-metrics}.
\end{proof}

\section{Importance sampling, proof of theorem
\ref{Theorem:approximation-algorithm-1}}
\label{Section:importance-sampling}

After we decomposed $EMD_X$ into $\sum_{v \in U_T}EEMD_v$, we'll use importance
sampling to choose a element $v \in U_T$, as mentioned in
section \ref{Section:overview}. The probability that $v$ is selected should be
roughly the ratio of $EEMD_v$ and $EMD_T$.  The following lemma from
\cite{Ind07} tells us how many samples are needed.

\begin{lemma}[\cite{Ind07}] \label{Lemma:importance-sampling}
Let $Z_1,\cdots,Z_s \geq 0, Z = \sum_{i}Z_i$ and $q_i = Z_i / Z$.
Let $p_1, p_2, \cdots, p_s$ be some numbers satisfying $p_i \geq q_i / \gamma$
and $\sum_{i=1}^s p_i = 1$, for some $\gamma \geq 1$. Consider a random
variable $S$ such that $\Pr[S = Z_i/p_i] = p_i$; note that $E[S] = Z$. Then,
$$\Pr\left[\card{\E_{i \in [M]}S_i - Z} > 0.5Z\right] \leq
e^{-\Omega(M/\gamma)},$$
where $S_1, S_2, \cdots, S_M$ are independent copies of $S$.
\end{lemma}
\begin{proof} 
With probability 1,
$$0 \leq S \leq \max_{i\in [s]}\frac{Z_i}{p_i} = \max_{i \in [s]}\frac{Z
q_i}{p_i}\leq \gamma Z$$
By Chernoff bound, we have $$\Pr\left[\card{\E_{i \in [M]}S_i - Z} > 0.5Z\right]
\leq
e^{-\Omega(M/\gamma)}.$$
\end{proof}

To apply lemma \ref{Lemma:importance-sampling}, we need a way to approximate
each $EEMD_v$ within a reasonable ratio.  Embedding the metric $d_v$ into
distribution of trees can
give an estimation that is always at least $EEMD_v$, and at most a
logarithmic factor times $EEMD_v$ in expectation(see \cite{Cha02} and
\cite{FRT03}). However,
to allow a good estimation, the number of trees in the support should be
$\Omega(\card{\Lambda_v})$. Then, we need roughly $\sum_{v \in
U_T}\card{\Lambda_v}^2 = O(n\max_{v}\card{\Lambda_v})$ running time for the
importance sampling, as opposed to $\tilde O(n)$ stated in the theorem. We can
do slightly better by using the min of the EMDs, instead of the average.

\begin{lemma}
\label{Lemma:embedding-into-min-EMDs}
Let $(Y, d_Y)$ be a metric space such that $\card{Y} \leq n$, and $\mathcal{T}$
be a distribution of dominating trees that approximate $(Y, d_Y)$ up to
$O(\log n)$ factor, i.e, $\E_{\tau \in T}d_\tau(p, q)\leq O(\log n)d_Y(p,
q)$. If we randomly choose $s = O(\log n)$ trees $\tau_1, \tau_2, \cdots,
\tau_s$ from $\mathcal{T}$, then for every pair of distributions
$\mu, \nu$ over $Y$, 
$$\min_{i \in [s]}EMD_{\tau_i}(\mu, \nu) \leq O(\log n) EMD_Y(\mu, \nu)$$ with
probability at least $1 - 1/n$.
\end{lemma}
\begin{proof}
For 1 tree $\tau$, with probability $1/2$, $EMD_{\tau}(\mu, \nu) \leq
O(\log n)EMD_Y(\mu, \nu)$. So, with probability at least $1 - 1/n$, there
exists a $i \in [s]$ satisfying $EMD_{\tau_i}(\mu, \nu) \leq O(\log
n)EMD_Y(\mu, \nu)$.
\end{proof}

Now we can proceed to prove theorem \ref{Theorem:approximation-algorithm-1}.
\begin{proof}[Proof sketch of theorem \ref{Theorem:approximation-algorithm-1}]
In the preprocessing, we run algorithm
\ref{Algorithm:embedding-into-slhst} to construct a SLHST $T$, with $\e$ to be
decided shortly. For each metric $(\Lambda_v, d_v)$ associated with
an inner vertex $v \in U_T$, we choose $s =
O(\log n)$ dominating trees $\tau_{v,1},
\tau_{v,2}, \cdots, \tau_{v, s}$ from the distribution of dominating trees that
$O(\log n)$-approximates $d_v$. The data structure includes $T$, all
$\tau_{v,i}$s and all sub-metrics of $X$ restricted to $\Lambda_v$ for $v\in
U_T$.

 The time to sample a SLHST is $\tilde O(n^2)$ and the time to compute all
$\tau_{v, i}$ is still $\tilde O(n^2)$, implying the preprocessing time is
$\tilde O(n^2)$. 

The size of the data structure is $\tilde
O(\sum_{v \in U_T}\card{\Lambda_v}^2) = n^{1 + O(\e)}$. The $\card{\Lambda_v}^2$
space is for small metrics $d_v$s and
sub-metrics of $d_X$ restricted to $\Lambda_v$s.

Let $\mu, \nu \in \mathcal{P}_X$ be the two distributions in the query. 

 With probability at least 0.9, the following two things happen:
\begin{enumerate}
\item  $EMD_T = \sum_{v \in U_T}EEMD_v$ approximates
$EMD_X$ with approximation ratio $O(\alpha/\e)$;
 \item For every $v \in U_T$, $\min_{i \in [s]}EEMD_{\tau_{i,v}}(\hat \mu_v,
\hat \nu_v) \leq O(\log n)EEMD_v(\hat \mu_v, \hat \nu_v)$;
\end{enumerate}
by lemma \ref{Lemma:emd-of-dominating-SLHSTs} and lemma
\ref{Lemma:embedding-into-min-EMDs}.

Assuming the above two things happen, we can do the importance sampling,
using $\min_{i
\in [s]}EEMD_{\tau_{i,v}}(\hat \mu_v, \hat \nu_v)$ as the weight for
$v$. The probability that $v$ is selected
in the importance sampling is at least $\displaystyle \frac{EEMD_v(\hat \mu_v,
\hat \nu_v)}{O(\log n)EMD_T}$. So, by lemma \ref{Lemma:importance-sampling}, we
only need $O(\log n)$ samples to approximate $EMD_T(\mu, \nu)$ up to a constant
factor. This total time for importance sampling is $\tilde O(n)$, since
$\sum_{v \in U_T}\card{\Lambda_v} = \tilde O(n)$.

Let $v$ be a vertex selected by the importance sampling. We compute
$EEMD_v(\hat \mu_v, \hat \nu_v)$ by using standard Hungarian algorithm.
$\card{\Lambda_v}$ is at most $n^{O(\e)}$, for small enough $\e$, computing
$EEMD_v$ takes $o(n)$ time.

In all, with preprocessing time $O(n^2)$, data structure of size $\tilde
O(n^{1 + \e})$, we can approximate $EMD_X(\mu, \nu)$ up to a $O(\alpha_X/\e)$
factor in $\tilde O(n)$ time.  Notice that we can remove the $O(.)$ notion in
the exponent, by losing a constant factor in the approximation ratio.
\end{proof}

\section{Binary importance sampling, proofs of theorem
\ref{Theorem:encoding-scheme}, \ref{Theorem:communication-protocol} and
\ref{Theorem:distance-oracle}}
\label{Section:binary-importance-sampling}

In this section, we describe an encoding scheme, where each distribution $\mu$
is encoded into a linear code $F(\mu)$ and $EMD_X(\mu, \nu)$ can be
approximated using sub-linear time, when $F(\mu)$ and $F(\nu)$ are given. 

Recall that, in the algorithm for theorem
\ref{Theorem:approximation-algorithm-1}, we need to do the importance sampling,
where we estimate $EEMD_v$ for every $v \in T$. To design a sub-linear
estimation time algorithm, we must avoid estimating $EEMD_v$ for every $v$.

This can be done by using a binary sampling method where, we maintain a set
$S$ of
possible outputs, initially $U_T$ and during each iteration, set $S$ is divided
into 2 equal subsets, and replaced with one of the subset with probability
proportional to the weight of the subset. The process repeats $O(\log n)$ times,
until
there's only 1 element left in $S$, which is the output of the importance
sampling.

The above method requires a good estimation for the total weight of a subset. In
section \ref{Section:importance-sampling}, we used the min of $l_1$ norms as the
estimation for the weight of an element.  Then the total weight of a subset is
the sum of mins, which seems hard to estimate. We'll use a different estimation
: the sum of
$l_1$ norms, which is still a $l_1$ norm. The total weight of a subset is again
the sum of $l_1$ norms, equivalent to a $l_1$ norm. There are good sketching
schemes for $l_1$ metric, for example, \cite{KOR98} used a random linear map
$g:l_1 \to \mathbb{R}^k$ as the sketching of a $l_1$ vector, such that the $L_1$
norm of $x$ is approximated by the median of $\card{g_1(x)}, \card{g_2(x)},
\cdots, \card{g_k(x)}$, where $g_i$ is the $i$-th coordinate of $g$. In the
sketch function, each $g_i$ is a linear function of $x$, where the linear
coefficients are chosen from the Cauchy distribution of scale parameter 1. i.e
$g_i = \sum_{j}\alpha_j x_j, \alpha_j \sim \mathbb{C}(0, 1)$.
The value
$k$ determines how well the norm is approximated, as in lemma
\ref{Lemma:sketch-size}. Before giving the lemma, we first define:

\begin{definition}[$\rho$-good] 
For some $0 < \rho < 0.1$, we say the sketch function $g$ is \emph{$\rho$-good}
for a fixed $x$, if 
\begin{eqnarray*}
(1 - \rho)|x|_1 &\leq& median(\card{g_1(x)},
\card{g_2(x)},
\cdots, \card{g_k(x)})\\
&\leq& (1 + \rho)|x|_1
\end{eqnarray*}
where $g_1, g_2, \cdots, g_k$ are $k$ coordinates of $g$.
\end{definition}

\begin{lemma}
\label{Lemma:sketch-size}
Let $g:l_1 \to \mathbb{R}^k$ be the $l_1$ sketch function, with $k = c/\rho^2$
for some $0 < \rho < 0.1$ and integer $c$. Then, for a fixed $x\in l_1$,
 $g$ is $\rho$-good with probability at least $1 - e^{-\Omega(c)}$.
\end{lemma}

We leave the proof of lemma \ref{Lemma:sketch-size} to the appendix. For a set
of $N$ elements, we fix a binary sampling tree, where each node is a subset
of $[N]$. A set
is equal to the union of its two child sets of half the size. The root of the
tree is $[N]$, and the $N$ leaves are $\set{1}, \set{2}, \cdots,
\set{N}$.  Let $\mathcal{S}_N$ be the family of all subsets of $[N]$ in this
binary sampling tree; notice that $\card{\mathcal{S}_N} = O(N)$, and
$\sum_{S\in \mathcal{S}_N}|S| = \tilde O(N)$. Our encoding
$f:(l_1)^N \to \mathbb{R}^{O(Nk)}$ is defined as 
\begin{equation}
\label{Equation:f}
f(x_1, x_2, \cdots, x_N) =
\Directprod_{S \in \mathcal{S}_N}g\left(\Directprod_{i \in S}x_i\right). 
\end{equation}

Algorithm \ref{Algorithm:binary-importance-sample} is the binary importance
sampling algorithm. It takes the encoding $f$ for $N$ vectors in $l_1$,  and
selects a number $i \in [N]$, with probability supposed to be
$\card{x_i}/\sum_{j\in [N]}\card{x_j}$.

\begin{algorithm}
 \caption{$b\_import\_sample(f(x_1, x_2, \cdots, x_N))$}
 \label{Algorithm:binary-importance-sample}
 \textbf{Input:} Encoding $f$ for $N$ vectors $x_1, x_2, \cdots, x_N$;\\
 \textbf{Output:} A pair $(t, P_t)$, where $t \in [N]$ and $P_t$ is the
probability that the algorithm selects $t$;
 \begin{algorithmic}[1]
 \STATE $S\leftarrow [N], P \leftarrow 1$;
 \WHILE {$\card{S} > 1$}
   \STATE Let $S_1$ and $S_2$ be the two children of $S$ in the binary sampling
tree;
    \STATE Extract $\hat g^1 = g\left(\Directprod_{j \in S_1}x_j\right)$ and
$\hat g^2 = g\left(\Directprod_{j \in S_2}x_j\right)$ from the encoding $f$;
    \STATE $W_1 \leftarrow median\left(\card{\hat g^1_1}, \card{\hat g^1_2},
\cdots, \card{\hat g^1_k}\right)$ and $W_2 \leftarrow median\left(\card{\hat
g^2_1}, \card{\hat g^2_2}, \cdots, \card{\hat g^2_k}\right)$;
   \STATE Let $i$ be 1 with probability $\frac{W_1}{W_1 + W_2}$ and 2 with
probability $\frac{W_2}{W_1 + W_2}$;
   \STATE $S \leftarrow S_i, P \leftarrow P \times \frac{W_i}{W_1 + W_2}$;
 \ENDWHILE
 \RETURN $(t, P)$ where $t$ is the unique element in $S$;
 \end{algorithmic}
\end{algorithm}

\begin{claim}
$P_t$ is the probability that algorithm \ref{Algorithm:binary-importance-sample}
selects $t$, as stated.
\end{claim}

\begin{definition}[successful]
\label{Definition:successful}
 We say a encoding $f$ is \emph{successful} for fixed $x_1, x_2, \cdots, x_N$, 
if algorithm \ref{Algorithm:binary-importance-sample} selects $t$ with
probability at least $0.5\card{x_t}_1/\card{\Directprod_{i=1}^{N}x_i}_1$.
\end{definition}

\begin{lemma}
\label{Lemma:success-probability}
Let $x_1, x_2, \cdots, x_N$ be fixed. If all occurrences
of $g$ in the encoding $f$ are $1/(10\log N)$-good, $f$ is successful.
\end{lemma}

\begin{proof}
If all occurrences of $g$ in $f$ are $1/(10\log N)$-good, the median of absolute
values of the sketch will always give the $l_1$ norm of the pre-image, up to a
factor of $1\pm 1/(10\log N)$. Thus the probability that algorithm
\ref{Algorithm:binary-importance-sample} selects $t$ is at least $$\left(\frac{1
-
\frac{1}{10\log N}}{1 + \frac{1}{10\log N}}\right)^{\log
N}\frac{\card{x_t}_1}{\card{\Directprod_{i=1}^{N}x_i}_1} > 0.5
\frac{\card{x_t}_1}{\card{\Directprod_{i=1}^{N}x_i}_1},$$ where the exponent
$\log N$ comes from the depth of the binary sampling tree.
 By definition
\ref{Definition:successful}, $f$ is successful. 
\end{proof}

Then we apply algorithm \ref{Algorithm:binary-importance-sample} to compute the
EMD
over a SLHST metric. As shown in lemma \ref{Lemma:emd-to-sum-emds}, EMD
over a SLHST $T$ is the sum of $EEMD_v$s over all $v\in U_T$. Each underlying
metric $d_v$ can be embedded into distribution of dominating trees, with
distortion logarithmic in the size of the metric. In order to make the
preprocessing efficient, we need a small set of dominating trees.
$\cite{CCG98}$ gives exactly what we want : 

\begin{theorem}[\cite{CCG98}]
Given a metric $(X, d_X)$ with $|X| = n$, there is a set of $O(n\log n)$
dominating trees and a probability on them, such that the metric $d_X$ is
approximated by the set of dominating tree metrics, with distortion $O(\log
n\log\log n)$. Moreover, there is a polynomial algorithm which gives the set
and the probability.
\end{theorem}

Then, by lemma \ref{Lemma:emd-of-dominating-metrics} (or \cite{Cha02}), the
EEMD over $(v, d_v)$ can be embedded into the average EEMD of $O(|\Lambda_v|\log
|\Lambda_v|)$ trees, which is equivalent to the norm of a
$(|\Lambda_v|^2\log |\Lambda_v|)$ dimensional $L_1$, with distortion $O(\log
|\Lambda_v| \log\log
|\Lambda_v|)$. We use $h_v : \mathbb{R}^{\Lambda(v)} \to \mathbb{R}^{\tilde
O(\card{\Lambda(v)}^2)}$ to
denote the linear function whose $l_1$ norm approximates $EEMD_v$. 

For a SLHST $T$, let's list all the vertices in $U_T$ : $v_1, v_2, \cdots,
v_N$. Define : 
\begin{equation}
\label{Equation:FT}
F_T(\mu) = f\left(h_{v_1}(\hat \mu_{v_1}),\cdots,
h_{v_N}(\hat \mu_{v_N})\right)\oplus \Directprod_{i \in
[N]}\hat\mu_{v_i}\\
\end{equation}

Then, algorithm \ref{Algorithm:approximate-emd} shows how to approximate the
EMD in sub linear time, if the encodings are given.

\begin{algorithm}
\caption{$approx\_EMD(T, F_T(\mu), F_T(\nu))$}
\label{Algorithm:approximate-emd}
\begin{algorithmic}[1]
 \STATE $Sum \leftarrow 0$;
 \STATE Extract $f_\mu = f\left(h_{v_1}(\hat \mu_{v_1}),\cdots,
h_{v_N}(\hat \mu_{v_N})\right)$ from $F_T(\mu)$, and $f_\nu =
f\left(h_{v_1}(\hat \nu_{v_1}),\cdots,
h_{v_N}(\hat \nu_{v_N})\right)$ from $F_T(\nu)$;
 \FOR{ $i \leftarrow 1, 2, \cdots, \lceil \log^2 n \rceil$ }
 \STATE $(t, P) = b\_import\_sample(f_\mu - f_\nu)$;
 \STATE Extract $\hat \mu_{v_t}$ from $F_T(\mu)$  and
$\hat \nu_{v_t}$ from $F_T(\nu)$;
  \STATE $Sum \leftarrow Sum + EEMD_{v_t}(\hat \mu_{v_t},
\hat \nu_{v_t})/P$;
 \ENDFOR
\RETURN $Sum/\lceil \log^2 n\rceil$.
\end{algorithmic}
\end{algorithm}

\begin{lemma}
\label{Lemma:approximate-emd-over-slhst}
If all occurrences of $f$ in $F_T(\mu - \nu)
= F_T(\mu) - F_T(\nu)$ are successful, 
algorithm \ref{Algorithm:approximate-emd} outputs
a number between $0.5 EMD_T(\mu, \nu)$ and $1.5 EMD_T(\mu,
\nu)$ with probability at least 0.9.
\end{lemma}

\begin{proof}
If all occurrence of $f$ in $F_T(\mu - \nu)$ are successful, algorithm
\ref{Algorithm:binary-importance-sample} chooses $t$ with probability at least
\begin{eqnarray*}
& &\frac{0.5h_t(\hat \mu_t, \hat \mu_t)}{\sum_{i\in [N]}h_{v_i}(\hat\mu_i,
\hat\nu_i)}\\
&\geq& \frac{0.5 EEMD_t(\hat\mu_t, \hat\nu_t)}{O(\log n\log\log
n)\sum_{i \in [N]}EEMD_i(\hat\mu_i, \hat\nu_i)}\\
&=&\frac{1}{O(\log n\log\log n)}\frac{EEMD_t(\hat\mu_t, \hat\nu_t)}{EMD_T(\mu,
\nu)}
\end{eqnarray*}

By lemma \ref{Lemma:importance-sampling}, the probability that the algorithm
outputs a number between $0.5EMD_T(\mu, \nu)$ and $1.5EMD_T(\mu, \nu)$ is at
least 
$$1 - e^{-\Omega(\log^2n/\log n\log\log n)} \geq 0.9.$$
\end{proof}

\begin{proof}[Proof sketch of theorem \ref{Theorem:encoding-scheme},
\ref{Theorem:communication-protocol} and
\ref{Theorem:distance-oracle}]
We first prove theorem \ref{Theorem:encoding-scheme}, and then show how this
it implies theorem \ref{Theorem:communication-protocol} and
\ref{Theorem:distance-oracle}.

The preprocessing stage is almost the same as the algorithm for theorem
\ref{Theorem:approximation-algorithm-1}, except that we choose $s =
O(|\Lambda_v|\log |\Lambda_v|)$ dominating trees $\tau_{v, 1}, \tau_{v,2},
\cdots, \tau_{v,s}$ for $v$ using the algorithm in \cite{CCG98}.  The
preprocessing time is $\tilde O(n^2)$, and the size of the data structure is
$O(n^{1 + \e})$.

We can suppose we have $EMD_T$ approximates $EMD_X$ with $O(\alpha_X/\e)$
distortion, because this happens with high probability.

Then we compute the encoding $F_T(\mu)$ and $F_T(\nu)$ using formulas
(\ref{Equation:f}) and (\ref{Equation:FT}). The time to compute $f$ is $k$
times the dimension of the input, which is $\sum_{v \in U_T}\tilde
O(\card{\Lambda_v}^2) = O(n^{1 + \e})$ in $F_T(\mu)$. The second part of
$F_T(\mu)$ can be ignored. The time to compute the encoding is $\tilde O(kn^{1 +
\e})$. The size of the encoding is $O(Nk) = \tilde O(nk)$.

If we set $k$ to $O(\log^3 n)$,  then by lemma \ref{Lemma:sketch-size}, every
occurrence of $g$ is $O(1/\log n)$-good with probability $1-O(1/n^4)$. Totally,
there are $\tilde O(|U_T|) = \tilde O(n)$ occurrences of $g$ in $F_T(\mu -
\nu)$, thus, with probability $1-O(1/n^3)$, every occurrence of $g$ is
$O(1/\log n)$-good, which implies, by lemma \ref{Lemma:success-probability},
$f(h_{v_1}(\hat\mu_{v_1}), \cdots, h_{v_N}(\hat\mu_{v_N}))$ is successful. This
finally implies algorithm will output a number between $0.5EMD_T(\mu, \nu)$
and $1.5EMD_T(\mu, \nu)$ with probability at least 0.9, by lemma
\ref{Lemma:approximate-emd-over-slhst}.  The number
$O(\alpha_X/\e)$-approximates $EMD_X(\mu, \nu)$. 

Since $k = O(\log^3 n)$, the time need to compute an encoding is $\tilde
O(kn^{1 + \e}) = \tilde O(n^{1 + \e})$, and the size of an encoding is $\tilde
O(nk) = \tilde O(n)$. The binary importance sampling takes time $\tilde O(1)$,
The bottleneck of algorithm \ref{Algorithm:approximate-emd} is
computing the $EEMD_v$s, which takes total time $\tilde n^{O(\e)})$. In
algorithm
\ref{Algorithm:approximate-emd}, we don't really compute $F_T(\mu - \nu)$.
As the encoding is linear, whenever we want to read a number from $F_T(\mu
- \nu)$, we read the numbers in $F_T(\mu)$ and $F_T(\nu)$ and then do the
subtraction.

Theorem \ref{Theorem:encoding-scheme} immediately implies theorem
\ref{Theorem:communication-protocol}. $Alice$ computes
$F(\mu)$ and $Bob$ computes $F(\nu)$, then, they communicate to simulate
algorithm \ref{Algorithm:approximate-emd}. The communication complexity is
$O(n^\e)$, with approximation ratio $\alpha_X/\e$.

For theorem \ref{Theorem:distance-oracle}, we can store the encodings for all
the $s$ distributions in the preprocessing stage and when a query comes,
run algorithm \ref{Algorithm:approximate-emd}. We can run the preprocessing
stage $O(\log s)$ times independently and store $O(\log s)$ data structures so
that with high probability, the EMD between every pair is reserved.
\end{proof}

\section{Conclusions}
We have demonstrated a almost linear algorithm for estimating EMD over doubling
metrics up to a constant factor. We also derived an encoding-based
sub-linear time constant approximation algorithm, which may have further
applications.

An interesting direction to purse would be to find a sketching scheme for EMD
over doubling metrics. \cite{ABIW09}'s sketching scheme benefits from the fact
that all the small grids are the same. While in our case, the small metrics are
different. If we can embed the EMD over the small metrics to any uniform
normed space up to a constant factor, we can use the technique in \cite{ABIW09}
to derive a sketch scheme.

\bibliographystyle{plain}
\bibliography{reflist}

\appendix

\section{Proof of lemma \ref{Lemma:sketch-size}}
\begin{proof}
For a fixed $x$ and $i\in [k]$, the distribution of $g_i(x)$ is $\mathbb{C}(0,
|x|_1)$. Let $\gamma = |x|_1$, and
\begin{eqnarray*}
 f(t) = \left\{
\begin{array}{cc}
 0 & \mbox{if } t < 0\\
 \frac2\pi\frac{\gamma}{t^2 + \gamma^2} & \mbox{if }t \geq 0
\end{array}
\right.
\end{eqnarray*}
be the probability density function $|g_i|$s. $F$ is the correspondent
cumulative function :
\begin{eqnarray*}
 F(t) = \left\{
\begin{array}{cc}
 0 & \mbox{if } t < 0\\
 \frac{2}{\pi}\arctan\left(\frac{t}{\gamma}\right)& \mbox{if }t
\geq 0
\end{array}
\right.
\end{eqnarray*}

For some $0 < \delta < 0.1$, define $Y_i = 1$ if $g_i < F^{-1}(1/2 - \delta)$
or $g_i > F^{-1}(1/2 + \delta)$ and 0 otherwise. Using Chernoff bound, we have
$$\Pr\left[\sum_{i=1}^k Y_i \geq k/2 \right] \leq e^{-\delta^2k/8}$$
If $k = c/\delta^2$, the above probability is at most $e^{-\Omega(c)}$. For
small enough $\delta$, $F^{-1}(1/2 - \delta) = (1 - \Theta(\delta))\gamma$ and
$F^{-1}(1/2 + \delta) = (1 + \Theta(\delta))\gamma$. $\sum_{i = 1}^k Y_i < n/2$
implies the absolute median falls between $(1 - O(\delta))\gamma$ and $(1 +
O(\delta))\gamma$.
\end{proof}

\end{document}